\documentclass[letterpaper,11pt]{article}
\pdfoutput=1
\usepackage[letterpaper,margin=1in]{geometry}

\usepackage[utf8]{inputenc}
\usepackage{amsmath,amsthm,amssymb}
\usepackage[margin=1in]{geometry}
\usepackage{xcolor,xspace}
\usepackage{cleveref}
\usepackage{paralist}

\newtheorem{theorem}{Theorem}[section]

\newtheorem{lemma}[theorem]{Lemma}
\newtheorem{claim}{Claim}
\newtheorem*{claim*}{Claim}
\newtheorem{remark}[theorem]{Remark}
\newtheorem{corollary}[theorem]{Corollary}

\newcommand{\LOCAL}{LOCAL\xspace}
\newcommand{\CONGEST}{CONGEST\xspace}
\newcommand{\eps}{\varepsilon}
\newcommand{\logstar}{\ensuremath{\log^*}}

\DeclareMathOperator{\poly}{poly}
\DeclareMathOperator{\polylog}{poly\log}
\DeclareMathOperator{\quasipolylog}{quasi\polylog}

\newcommand{\kz}{\ensuremath{k}}
\newcommand{\ka}{\ensuremath{k'}}
\newcommand{\tz}{\ensuremath{\tau}}
\newcommand{\ta}{\ensuremath{\tau'}}
\newcommand{\Pow}{\ensuremath{\mathcal{P}}}
\newcommand{\PP}{\mathfrak{P}}

\newcommand{\myemail}[1]{\,$\cdot$\, {\small #1}}
\newcommand{\myaff}[1]{\,$\cdot$\, {\small #1}\par\smallskip}

\newenvironment{mycover}
{\list{}{\listparindent 0pt
        \itemindent    \listparindent
        \leftmargin    1cm
        \rightmargin   0.5cm
        \parsep        0pt}%
    \raggedright
    \item\relax}
{\endlist}

\begin{document}
\begin{mycover}
{\huge\bfseries\boldmath Local Conflict Coloring Revisited:\\ Linial for Lists \par}
\bigskip
\bigskip
\bigskip

\textbf{Yannic Maus}
\myemail{yannic.maus@cs.technion.ac.il}
\myaff{Technion}

\textbf{Tigran Tonoyan}
\myemail{ttonoyan@gmail.com}
\myaff{Technion}
\end{mycover}
\bigskip

\begin{abstract}
Linial's famous color reduction algorithm reduces a given $m$-coloring of a graph with maximum degree $\Delta$ to a $O(\Delta^2\log m)$-coloring, in a single round in the \LOCAL model. We show a similar result when nodes are restricted to choose their color from a list of allowed colors: given an $m$-coloring in a directed graph of maximum outdegree $\beta$, if every node has a list of size $\Omega(\beta^2 (\log \beta+\log\log m + \log \log |\mathcal{C}|))$ from a color space $\mathcal{C}$ then they can select a color in two rounds in the \LOCAL model. Moreover, the communication of a node essentially consists of sending its list to the neighbors.
This is obtained as part of a framework that also contains Linial's color reduction (with an alternative proof) as a special case. Our result also leads to a \emph{defective list coloring} algorithm.
As a corollary, we improve the state-of-the-art \emph{truly local} $(deg+1)$-list coloring algorithm from Barenboim et al. [PODC'18] by slightly reducing the runtime to $O(\sqrt{\Delta\log\Delta})+\log^* n$ and significantly reducing the message size (from huge to roughly $\Delta$). 
Our techniques are inspired by the \emph{local conflict coloring} framework of Fraigniaud et al. [FOCS'16]. 
\end{abstract}

\thispagestyle{empty}
\setcounter{page}{0}
\newpage
\section{Introduction}\label{sec:intro}
\emph{Symmetry breaking problems}  are a cornerstone of distributed graph algorithms in the \LOCAL model.\footnote{In the \LOCAL model \cite{linial92, peleg00} a communication network is abstracted as an $n$-node graph $G=(V, E)$ with unique $O(\log n)$-bit identifiers. Communications happen in synchronous rounds. Per round, each node can send one (unbounded size) message to each of its neighbors. At the end, each node should know its own part of the output, e.g., its own color. In the \CONGEST model, there is a limitation of $O(\log n)$ bits per message.}  A central question in the area asks: \textit{How fast can these problems be solved in terms of the maximum degree $\Delta$, when the dependence on $n$ is as mild as $O(\logstar n)$?} \cite{barenboimelkin_book}. 
That is, we are looking for  \emph{truly local} algorithms, with complexity of the form $f(\Delta)+O(\logstar n)$.\footnote{Not all symmetry breaking problems admit such a runtime, e.g., $\Delta$-coloring and sinkless orientation \cite{brandt2016LLL}.} The $O(\logstar n)$  term is unavoidable due to the seminal lower bound by Linial \cite{linial92} that, via simple reductions, applies to most typical symmetry breaking problems. $(\Delta+1)$-\emph{vertex coloring} and \emph{maximal independent set (MIS)} are the key symmetry breaking problems. Both can be solved by simple centralized greedy algorithms (in particular they are always solvable), and even more importantly in a distributed context, any partial solution (e.g. a partial coloring) can be extended to a complete solution. The complexity of MIS is settled up to constant factors with $f(\Delta)=\Theta(\Delta)$, by the algorithm from \cite{BarenboimEK14} and a recent breakthrough lower bound by Balliu et al. \cite{FOCS19MIS}. 
In contrast, the complexity of vertex coloring, despite being among the most studied distributed graph problems \cite{barenboimelkin_book}, remains widely open, with the current best upper bound $f(\Delta)=\tilde{O}(\sqrt{\Delta})$ and lower bound $f(\Delta)=\Omega(1)$. 
While most work has focused on the $(\Delta+1)$-coloring problem, recent algorithms, e.g., \cite{barenboim16,FHK,BEG18,Kuhn20}, 
 rely  on the more general \emph{list coloring problem} as a subroutine, where each vertex $v$ of a graph $G$ has a list $L_v\subseteq \mathcal{C}$ of colors from a colorspace $\mathcal{C}$, and the objective is to compute a proper vertex coloring, but each vertex has to select a color from its list. Again, a natural case is the always solvable \emph{$(deg+1)$-list coloring problem}, where the list of each vertex is larger than its degree. 
Our paper contributes to the study of list coloring problems. To set the stage for our results, let us start with an overview of truly local coloring algorithms.
 
 In~\cite{linial92}, besides the mentioned lower bound, Linial also showed that  $O(\Delta^2)$-coloring can be done in $O(\logstar n)$ rounds. Szegedy and Vishwanathan \cite{SzegedyV93} improved the runtime to $\frac{1}{2}\logstar n+O(1)$ rounds, and  showed that in additional $O(\Delta\cdot \log \Delta)$ rounds, the $O(\Delta^2)$-coloring could be reduced to $(\Delta+1)$-coloring. The latter result was rediscovered by Kuhn and Wattenhofer \cite{KuhnW06}. Barenboim, Elkin and Kuhn used \emph{defective coloring} for  partitioning the graph into low degree subgraphs and coloring in a divide-and-conquer fashion, and brought the complexity of $(\Delta+1)$-coloring down to $O(\Delta + \logstar n)$ \cite{BarenboimEK14}. A simpler algorithm with the same runtime but without using defective coloring was obtained recently in~\cite{BEG18}. All of these results also hold for  $(deg+1)$-list coloring, since given a $O(\Delta)$-coloring, one can use it as a ``schedule'' for computing a $(deg+1)$-list coloring in $O(\Delta)$ additional rounds. Meantime, two sub-linear in $\Delta$ algorithms were already published \cite{barenboim16,FHK}. They both used \emph{low outdegree colorings}, or \emph{arb-defective colorings}, introduced by Barenboim and Elkin in~\cite{barenboimE10}, for graph partitioning purposes.  The basic idea here is similar to the defective coloring approach, with the difference that the graph is partitioned into  directed subgraphs with low \emph{maximum outdegree}. In \cite{BEG18} they also improved and simplified the computation of low outdegree colorings, which led to improved runtime and simplification of that component in the mentioned sublinear algorithms.  As a result, the currently fastest algorithm in \CONGEST needs $O(\Delta^{3/4} +\logstar n)$ rounds (\cite{barenboim16}+\cite{BEG18}),\footnote{For more colors, i.e.,  $(1+\eps)\Delta$-coloring, \cite{barenboim16} gives runtime $O(\sqrt{\Delta}+\logstar n)$ in \CONGEST.} while the fastest one  
 in \LOCAL  needs $O(\sqrt{\Delta\log\Delta}\logstar \Delta+\logstar n)$ rounds (\cite{FHK}+\cite{BEG18}). 

Let us take a better look at the latter result, which is the closest to our paper. The algorithm consists of two main ingredients:
(1) an algorithm that partitions a given graph into $p=O(\Delta/\beta)$ subgraphs, each equipped with a $\beta=O(\sqrt{\Delta/\log \Delta})$-outdegree orientation, in $O(p)+\frac{1}{2}\logstar n$ rounds~\cite{BEG18}, (2) a list coloring subroutine that gives rise to the following~\cite{FHK}:
\begin{theorem}[\cite{FHK}]
\label{thm:FHK}
In a directed graph with max. degree $\Delta$, max. outdegree $\beta$, and an input $m$-coloring, list coloring with lists $L_v$ of size  $|L_v|\geq 10\beta^2\ln \Delta$ from any color space can be solved in $O(\logstar(\Delta+ m))$ rounds in \LOCAL.
\end{theorem}
The two ingredients are combined to give a $(deg+1)$-list coloring of the input graph $G$ \cite{FHK}. First, partition $G$ using (1), iterate through the $p$ directed subgraphs, and for each of them, let uncolored nodes refine their lists by removing colors taken by a neighbor, and use (2) to color nodes that still have sufficiently large lists ($\geq 10\beta^2\ln\Delta$). With a fine grained choice of $\beta$ it is ensured that every node $v$ with (refined) list size greater than 
$\Delta/2$ is colored. Thus, after iterating through all $p$ subgraphs, all uncolored nodes have list size at most $\Delta/2$, and because each vertex always has \emph{more colors in its list than uncolored neighbors}, the max degree of the graph induced by uncolored nodes is at most half  of that of $G$. Thus, we can apply the partition-then-list-color paradigm recursively to the subgraph induced by uncolored nodes to complete the coloring. 
The runtime is dominated by the first level of recursion.

The strength of the partitioning algorithm above is that it is conceptually simple and works with small messages. In contrast, the algorithm from Thm.~\ref{thm:FHK} is conceptually complicated and uses gigantic messages. This complication might be due to the generality of the addressed setting as, in fact, \cite{FHK} studies the more general local conflict coloring problem, and Thm.~\ref{thm:FHK} is only a special case. In  \emph{local conflict coloring}, each edge of the given graph is equipped with an arbitrary conflict relation between colors and this relation may vary across different edges. This framework is also leveraged to achieve the  colorspace size independence of Thm.~\ref{thm:FHK} (see the technical overview below).
It was not clear prior to our work whether the situation simplifies significantly if one restricts to ordinary list coloring because, even if the input to the algorithm in Thm.~\ref{thm:FHK} is a list coloring problem, 
the intermediate stages of the algorithm fall back to the more general local conflict coloring.

The overarching goal of our paper is providing deeper understanding of the remarkable framework of~\cite{FHK}, better connecting it with classic works in the area, and obtaining a simpler  list coloring algorithm that also uses smaller messages, by moderately sacrificing generality. 

\subsection{Our Contribution} 
Our main  result is a simple algorithm that yields the following theorem. \begin{theorem}[Linial for Lists]
\label{thm:2rounds}
In a directed graph with max. degree $\Delta$, max. outdegree $\beta$, and an input $m$-coloring, list coloring with lists $L_v$ from a color space $\mathcal{C}$ and of size  $|L_v|\geq l=4e\beta^2(4\log \beta+\log\log |\mathcal{C}|+\log\log m+8)$  can be solved in 2 rounds in  \LOCAL. Each node sends $l+1$ colors in the first round, and a $l/\beta^2$-bit message in the second.
\end{theorem}
The name ``Linial for Lists'' stems from the fact that Thm.~\ref{thm:2rounds} is a ``list version'' of one of the cornerstones of distributed graph coloring, Linial's color reduction, which says that an $m$-coloring can be reduced to a $(5\Delta^2\log m)$-coloring in a single round~\cite{linial92}. Moreover, our framework is itself a \emph{natural generalization of  Linial's approach of  cover-free set families}. Applied to equal lists, it yields an alternative proof of Linial's color reduction, 
 in the form of a \emph{greedy construction of cover-free families} (Linial proved their existence using the probabilistic method \cite{linial92}; he also used an alternative construction from~\cite{EFF85} via polynomials over finite fields, which however yields a weaker color reduction for $m\gg\Delta$)  (see Sec.~\ref{sec:Linial} and Sec.~\ref{sec:discussion}).

Compared with Thm.~\ref{thm:FHK}, we lose colorspace independence, and our algorithm does not extend to general local conflict coloring (although we use a kind of conflict coloring in the process). In exchange, we eliminate $\Delta$ from the bound on the list size, reduce the runtime to exactly 2 rounds, and dramatically reduce message size. This is achieved by a  non-trivial paradigm shift in the local conflict coloring framework (see the technical overview below).
The runtime cannot be reduced to $1$ round, due to a lower bound of \cite{SzegedyV93} (see Sec.~\ref{sec:discussion}).

Combining Thm.~\ref{thm:2rounds} with the partitioning algorithm of~\cite{BEG18} as outlined above gives us a $(deg+1)$-list coloring algorithm. Note that any improvement in the list size bound in Thm.~\ref{thm:2rounds} (with little increase in runtime) would yield a faster $(deg+1)$-list coloring algorithm. 
\begin{theorem}[$(deg+1)$-List Coloring]
\label{thm:mainListColoring}
In a graph with max. degree $\Delta$, $(deg+1)$-list coloring with lists $L_v\subseteq \mathcal{C}$ from a color space  of size $|\mathcal{C}|=2^{\poly(\Delta)}$ \footnote{We use the notation $\poly(X)=O(X^c)$, for an absolute constant $c$, and $\tilde{O}(X)=X\cdot \poly(\log X)$.} can be solved in $O\big(\sqrt{\Delta\log\Delta}\big)+\frac{1}{2}\cdot\logstar n$ rounds in \LOCAL.
Furthermore, each node only needs to broadcast to its neighbors a single non-\CONGEST message consisting of a subset of its list.
\end{theorem}
The bound on the color space size stems from the color space dependence in~\Cref{thm:2rounds}. 
As discussed in Sec.~\ref{sec:discussion}, it is possible to trade color space dependence with runtime in~\Cref{thm:2rounds}, which could improve or suppress the bound in~\Cref{thm:mainListColoring}. That, however, comes with the cost of having huge messages. 

\Cref{thm:mainListColoring} immediately provides the fastest known truly local $(\Delta+1)$-coloring algorithm in \LOCAL. Below we list further implications of our framework.  
\begin{itemize}
\item \textbf{\CONGEST (see Cor.~\ref{cor:deltaplus1}):} We obtain an improved $(\Delta+1)$-coloring algorithm in a \emph{low degree} regime in \CONGEST. In particular, if $\Delta=\tilde{O}(\log n)$ then $(\Delta+1)$-coloring (more generally, $(deg+1)$-list coloring with colorspace of size $|\mathcal{C}|=\poly(\Delta)$) can be solved in $\tilde{O}(\sqrt{\Delta})+\frac{1}{2}\cdot\logstar n$ rounds in \CONGEST.  
Generally, if one allows messages of size $B$, this runtime  holds for degree up to $\Delta=\tilde{O}(B)$. 
On the other hand, if $\Delta=\Omega(\log^{2+\eps}n)$, for an arbitrarily small constant $\eps>0$, an algorithm from recent work~\cite{Kuhn20}  achieves runtime ${O}(\sqrt{\Delta})$ in \CONGEST (if one recasts their dependency on $n$ as a $\Delta$-dependency). Thus, only for the  regime of $\Delta \in \Omega(\log^{1+\epsilon} n)\cap O(\log^{2+\epsilon} n)$ 
we do not have an algorithm with runtime $\tilde{O}(\sqrt{\Delta})$ in  \CONGEST (with the current best being $O(\Delta^{3/4}+\logstar n)$ due to~\cite{barenboim16}).

\item \textbf{Defective list coloring (see Thm.~\ref{thm:listDefective}):} 
Our framework extends to \emph{$d$-defective list coloring}, that is, list coloring where each node $v$ can have at most $d$ neighbors with the same colors as $v$: If lists are of size $\Omega\big((\Delta/(d+1))^2 \cdot (\log\Delta + \log\log|\mathcal{C}| + \log\log m)\big)$ we can compute a $d$-defective list coloring in $2$ rounds in \LOCAL. 
The result can be seen as the ``list variant'' of a defective coloring result in~\cite{Kuhn2009WeakColoring}. 
 While we are not aware of an immediate application, defective list coloring with a better ``colors vs. defect'' tradeoff ($d$ vs. $O(\Delta/d)$) for \emph{line graphs} has recently been used to obtain a \emph{edge-coloring} algorithm with complexity $\quasipolylog \Delta + O(\logstar n)$~\cite{BKO20}.

\item \textbf{$\Delta$-coloring:} The improvements obtained in Thm.~\ref{thm:mainListColoring} also imply respective improvements for several \emph{$\Delta$-coloring} algorithms that use $(\deg+1)$-list coloring as a subroutine~\cite{GHKM18}. 
\end{itemize}

\subsection{Technical Overview}
At their core, the proofs of Theorems~\ref{thm:FHK} and~\ref{thm:2rounds} 
are based on three important concepts: \emph{conflict coloring}, \emph{problem amplification} and \emph{0-round solvability}.  
A \emph{conflict coloring problem} is a list coloring problem where two colors can conflict even if they are not equal. The associated \emph{conflict degree} is the maximum number of conflicts per color a node can have. 
\emph{Problem amplification} transforms one conflict coloring problem instance into another, as follows: given an input to a problem $A$, each node computes its input to another problem $B$ (perhaps by exchanging information along the way), with the property that, 1. having a solution to $B$, a simple one round algorithm computes a solution to $A$, and 2. the list-size-to-conflict-degree ($l/d$) ratio of $B$ is larger than that of $A$. Note that the first property essentially determines the conflicts in $B$, and usually a color in $B$ is a set of colors in $A$. The importance of the second property stems from the concept of \emph{0-round solvability}:  an instance of a problem $B$ with large enough $l/d$ ratio can be solved in 0 rounds, i.e., with no communication. 

From here, the plan is simple: take problem $P_0$, which is the list coloring problem, recast it as a conflict coloring problem, and amplify it into problems $P_1,\dots,P_t$, so that $P_t$ is 0-round solvable. Then we can cascade down to a solution of problem $P_0$, in $t$ rounds. Crucially, in order to do the above, we need $P_0$ to have sufficiently large $l/d$ ratio to begin with (which explains the particular list size requirements in our theorems). The input $m$-coloring is used for tie-breaking in the 0-round solution of $P_t$.

In~\cite{FHK}, \emph{local conflict coloring} is the main problem type, where the conflict between two colors depends on who the colors belong to, i.e., two colors can conflict along one edge of the graph and not conflict on another one. Their framework allows solving any local conflict coloring problem, and by re-modeling a problem with an arbitrary colorspace via mapping each list to an interval $[1,l]$ of natural numbers, one can redefine local conflicts and ``forget'' about the real size of the colorspace (hence  colorspace independence).  When computing the input of $P_i$ (given $P_{i-1}$), in order to maintain manageable conflict degree, \emph{nodes exchange messages} to filter out colors in $P_i$ that cause too much conflict with any neighboring node.  These messages are huge (recall that a color in $P_i$ is a set of colors in $P_{i-1}$). Thus, the input to $P_i$ is usually the topology, $P_0$-lists and conflicts in the $i$-hop neighborhood of a node. The goal towards 0-round solvability is then to find a problem $P_t$ whose $l/d$ ratio is larger than the number of all $t$-hop neighborhood patterns (i.e., inputs). The complicated nature of the input to $P_t$ also makes the 0-round solvability proof rather conceptually involved. The number $t$ of problems required is about $3\log^*(m+\Delta)$.

Our framework, on the other hand, is based on special \emph{global conflict coloring} instances, where the conflict relation of two colors does not depend on the edge across which they are. This limits us to solving only ordinary list problems $P_0$.  Our key insight (see~\Cref{sec:zeroRound}), which sets Theorems~\ref{thm:FHK} and~\ref{thm:2rounds} apart, is that in our setting \emph{nodes do not need to communicate} for computing the input to problems $P_1,\dots,P_t$. To achieve this, we show that when forming the lists for $P_i$ from the input to $P_{i-1}$, it suffices to drop ``universally bad'' colors (sets of colors in $P_{i-1}$), whose absence is enough to ensure moderate conflict degree towards any (!) other node. We achieve this by crucially exploiting the symmetry of the particular conflict coloring problems arising from ordinary list coloring. 

Thus, the input of a node in $P_t$ is just its input in $P_0$. This makes the 0-round solution (of $P_t$) particularly simple. The only communication happens when we cascade down from a solution of $P_t$ to that of $P_0$. With $t=2$, we get our main theorem. Since here we have only two problems, the message size is limited (the first round is needed to \emph{learn the $P_0$-lists} of neighbors, while the second one consists of a   \emph{small auxiliary message}). Taking larger $t$ would reduce the requirement on the initial list size but increase message size (see \Cref{sec:discussion}). Since $t=2$ is sufficient for our applications, we limit our exposition to that case. Setting $t=1$ does not give anything non-trivial for list coloring, since the $l/d$ ratio is not large enough, but when all $P_0$-lists are equal, it gives \emph{an alternative proof of Linial's color reduction}  (\Cref{sec:Linial}).
In fact, $P_1$ is essentially the problem of finding a \emph{low intersecting set family}, 
which Linial's algorithm is based on, while $P_2$ is a ``higher-dimensional'' variant of it. Thus, at the core of our result there is an (offline) construction of certain set families over the given color space: given those, the algorithm is easy. This way, we believe our paper also provides a deeper insight into the framework  
of Thm.~\ref{thm:FHK}. Our result can also be seen as a bridge between the results of~\cite{FHK} and the recently popular concept of speedup (see \Cref{sec:discussion}).

\subsection{Further Related Work}
\label{ssec:furtherRelated}
Most results on distributed graph coloring until roughly 2013 are covered in the excellent monograph by Barenboim and Elkin \cite{barenboimelkin_book}. An overview of more recent work can be found in \cite{Kuhn20}. Due to the large volume of published work on distributed graph coloring, we limit this section to an informative treatment of a selected subset.
While we have covered most literature on truly local vertex coloring algorithms, there are many known algorithms that trade the high $\Delta$-dependence in the runtime with lower $n$-dependence. All \textbf{deterministic} algorithms in this category (for general input graphs) involve a $\Omega(\log n)$ factor. 
From the early 90s until very recently, the complexity of $(\deg+1)$-list coloring (and $(\Delta+1)$-coloring) in terms of $n$ was $2^{O(\sqrt{\log n})}$ \cite{awerbuch89,panconesi1992improved}, with  algorithms based on \emph{network decomposition} (into small diameter components). 
A recent breakthrough in network decomposition algorithms~\cite{RG19} 
brought the runtime of $(deg+1)$-list coloring down to $\polylog n$ in \LOCAL (it also applies to many other symmetry breaking problems; see \cite{RG19,SLOCAL17,FOCS18-derand}). A little later, \cite{BKM19} found a $\polylog n$ round \CONGEST algorithm. 

Historically, decompositions into subgraphs that are equipped with low outdegree orientations as used in our results, in \cite{FHK}, and in \cite{barenboim16} are closely related to the notion of arboricity. To the best of our knowledge, \cite{BE10sublog} was the first paper to introduce low out-degree orientations as a tool for distributed graph coloring. First, they showed that one can compute $O(a)$-outdegree orientations in graphs with arboricity $a$ in $O(\log n)$ rounds, and used it to devise several algorithms to color graphs with bounded arboricity. \cite{BE10sublog} is also the first paper to notice that the degree bound of $\Delta$ in Linial's color reduction can be replaced with a bound on the outdegree.
Then,  \cite{barenboimE10} devised methods to recursively partition into graphs with small arboricity yielding an $O(\log \Delta \log n )$-round algorithm for $O(\Delta^{1+\eps})$-coloring and an $O(\Delta^{\eps}\log n)$-round algorithm for $O(\Delta)$ coloring. 
Recently, this recursive technique was extended to $(\deg+1)$-list coloring, giving a $(2^{O(\sqrt{\log \Delta})}\log n)$-round algorithm~\cite{Kuhn20}; the runtime of \cite{Kuhn20} has a hidden dependence on the color space. While \cite{BE10sublog,barenboimE10,Kuhn20} have an inherent $O(\log n)$-factor in their runtime,  \cite{barenboim16} showed that one can  decompose a graph into small arboricity subgraphs (equipped with a small outdegree orientation) without inferring a $O(\log n)$ factor, yielding the first sublinear in $\Delta$ algorithm for $\Delta+1$ coloring. In the aftermath, \cite{BEG18} improved the runtime for computing the underlying decompositions (and also simplified the algorithm).  Thus, the best forms of our results, \cite{FHK} and \cite{barenboim16} are obtained by using \cite{BEG18} to compute  decompositions into subgraphs of small arboricity (equipped with small outdegree orientations). 

Note that our results,  \cite{FHK} and \cite{barenboim16} only require a bound on the outdegree of the subgraphs' orientations and are oblivious to their arboricity. While bounded outdegree in a graph with a given orientation implies bounded arboricity, computing a bounded outdegree orientation in a graph with bounded arboricity requires $\Omega(\log n)$ rounds, as shown in \cite{BE10sublog}.

Recent \textbf{randomized} coloring algorithms rely on the \emph{graph shattering} technique \cite{BEPSv3}. 
 In the \emph{shattering phase}, a randomized algorithm computes a partial coloring of the graph, after which every uncolored connected component of the graph has small size (say, $\polylog n$).
Then, in the \emph{post-shattering} phase, deterministic $(\deg+1)$-list coloring is applied on all uncolored components in parallel.
The runtime of the shattering phase has progressed from $O(\log \Delta)$ \cite{BEPSv3}, over $O(\sqrt{\log \Delta})$ \cite{HsinSu18} to $O(\logstar \Delta)$ \cite{chang2017optimal}. Combined with the 
$\polylog n$-round list coloring algorithm of \cite{RG19}, this gives the current best runtime $\poly\log\log n$, for $(\Delta+1)$-coloring \cite{chang2017optimal}, and $O(\log\Delta)+\poly\log\log n$, for $(deg+1)$-list coloring \cite{BEPSv3}.

While special graph classes are out of the scope of this paper, we mention the extensively studied case of distributed \textbf{edge coloring}.
Here, $\polylog n$-round algorithms were designed for progressively improving number of colors, from $(2+\eps)\Delta$ \cite{GS17,GHKMSU17} to $(2\Delta-1)$ \cite{FGK17,HarrisEdge19}, then  to $(1+\eps)\Delta$  \cite{GKMU17,HarrisEdge19,SuVu19}. 
The truly local complexity of $(2\Delta-1)$-edge coloring has improved from $O(\Delta)$  
\cite{panconesi-rizzi} to  $2^{O(\sqrt{\log \Delta})})$ \cite{Kuhn20} then to   $\quasipolylog \Delta$ \cite{BKO20} (in addition to $O(\logstar n)$). $O(\Delta^{1+\eps})$-edge colorings can be computed in $O(\log \Delta+\logstar n)$ rounds \cite{BE11_neighborhoodInd}.

 Little is known on coloring \textbf{lower bounds} (in contrast to other symmetry breaking problems, e.g., maximal matching, MIS or ruling sets \cite{kuhn16_jacm,FOCS19MIS,balliu2020ruling}). 
  Linial's $\Omega(\logstar n)$ lower bound is extended to randomized algorithms in \cite{Naor91}. 
 The deterministic bound has recently been re-proven in a  topological framework \cite{fraigniaud2020topology}. 
A $\Omega(\Delta^{1/3})$ lower bound for $O(\Delta)$-coloring holds in a weak variant of the LOCAL model  \cite{disc16_coloring}. 
 Several works 
 characterized coloring algorithms which can only spend a single communication round \cite{SzegedyV93,KuhnW06,disc16_coloring}. None of these results gives anything non-trivial for two rounds. 
Also, the \emph{speedup} technique (e.g.,  \cite{Brandt19speedup,brandt2016LLL,FOCS19MIS,br2020truly,balliu2020ruling,balliu2019classification}), which proved very successful for MIS lower bounds, is poorly understood for graph coloring. We briefly discuss the technique and its relation to our result in Sec.~\ref{sec:discussion}.
There are lower bounds for 
more restricted variants of coloring. 
There is a $\Omega(\log n)$ ($\Omega(\log\log n)$) lower bound for deterministic \cite{brandt2016LLL} (randomized \cite{chang16exponential}) $\Delta$-coloring, as well as for $(\Delta-2)$-defective $2$-coloring \cite{BHKOS19}. 
Further, \cite{greedycoloring}  provides a $\Omega(\log n/\log\log n)$ lower bound for greedy coloring. Similar bounds 
hold for coloring trees and bounded arboricity graphs with significantly fewer than $\Delta$ colors \cite{linial92,BE10sublog}.

\subsection{Roadmap}
\Cref{sec:zeroRounds} introduces our version of conflict coloring 
together with the 0-round solvability lemma. 
\Cref{sec:notition} defines the problems $P_0$ and $P_1$ and provides further notation. 
\Cref{sec:Linial} contains the first result of our framework: an alternative proof of Linial's algorithm. 
\Cref{thm:2rounds} (Linial for Lists) is proved in \Cref{sec:versionTwo}.
\Cref{thm:mainListColoring} ($(deg+1)$-list coloring) is proved in \Cref{sec:coloring}. 
\Cref{thm:listDefective} (Defective list coloring) is proved in \Cref{sec:listdefective}. 
We conclude with a discussion of the results and open problems in \Cref{sec:discussion}.

\section{Basic Setup and Linial's Color Reduction}

In this section, we first introduce the conflict coloring framework that is the basis of our algorithm, then we show how it quickly implies an alternative variant of Linial's color reduction algorithm.
For a set $S$ and an integer $k\ge 0$, let $\Pow(S)$ and ${S\choose k}$ denote the set of all subsets and all size-$k$ subsets of $S$, respectively. For a map $f$ we use $f^{(i)}$ to denote the $i$-fold application of $f$, e.g., $\Pow^{(2)}(S)=\Pow(\Pow(S))$.

\subsection{Global Conflict Coloring}
\label{sec:zeroRounds}

 A \emph{list family} $\mathcal{F}\subseteq \Pow(\mathcal{C})$  is a set of subsets of a color space $\mathcal{C}$. Given a \emph{symmetric conflict relation} ${\mathcal{R}}\subseteq \{\{c,c'\} \mid c,c'\in \mathcal{C}\}$, the \emph{conflict degree} of a family  $\mathcal{F}$ in ${\mathcal{R}}$ is  the maximum number of colors in a list $L$ that conflict with a single color in a list $L'$ (possibly same as $L$), i.e.\ , $d_{\mathcal{R}}(\mathcal{F})=\max_{L,L'\in \mathcal{F}, c\in L}|\{c'\in L' \mid \{c,c'\}\in {\mathcal{R}}\}|$.
An instance $\PP=(\mathcal{C},{\mathcal{R}},\mathcal{F}, \mathcal{L})$ of the \emph{global conflict coloring problem} 
on the graph $G$ is given\footnote{Formally, $G$ is also part of the problem, but we omit it since it is always clear from the context. These definitions crucially differ from \underline{local} conflict coloring in~\cite{FHK}, where a pair of colors can conflict along one edge and not conflict along another. } by a color space $\mathcal{C}$, a symmetric conflict relation ${\mathcal{R}}$ on $\mathcal{C}$, a list family $\mathcal{F}$, and an assignment $\mathcal{L}:V\rightarrow \mathcal{F}$ of lists $\mathcal{L}(v)\in \mathcal{F}$ of colors to each vertex $v$. 
The goal is to assign each vertex a color from its list such that no pair of neighboring vertices get conflicting colors $\{c,c'\}\in {\mathcal{R}}$. 
The \emph{conflict degree} of $\PP$ is $d_{\mathcal{R}}(\mathcal{F})$. Note that the conflict degree \emph{does not depend on $G$ or $\mathcal{L}$}.

\begin{lemma}[Zero Round Solution]\label{l:zeroTypes} 
An instance $(\mathcal{C},{\mathcal{R}},\mathcal{F}, \mathcal{L})$ of the conflict coloring problem on a graph $G$ can be solved without communication if $G$ is $m$-colored, $m,{\mathcal{R}},\mathcal{F}$ are globally known, and every list in $\mathcal{F}$ has size at least 
$l > m\cdot |\mathcal{F}| \cdot d_{\mathcal{R}}(\mathcal{F})~$.
\end{lemma}
\begin{proof}
Every vertex $v$ has a \emph{type} $(\psi_v, \mathcal{L}(v))\in [m]\times \mathcal{F}$, which is uniquely determined by its input color $\psi_v$ and list $\mathcal{L}(v)$. Note that adjacent vertices have distinct types, and there are $t= m\cdot |\mathcal{F}|$ (globally known) possible types. Below, we show how to greedily assign each type a color from its list s.t. different types get non-conflicting colors. 
The conflict coloring problem is then solved by running this algorithm locally and consistently by all vertices, where each vertex gets the color assigned to its type.

Let $\{T_i=(m_i,L_i)\}_{i=1}^t$  be a fixed ordering of $[m]\times \mathcal{F}$.
Assign  $T_1$ a color $\phi(T_1)\in L_1$ arbitrarily.  For any $i\ge 1$, given the colors $\phi(T_1),\dots,\phi(T_{i})$ of preceding types, 
assign $T_{i+1}$ a color from $L_{i+1}$ that does not conflict with $\phi(T_1),\dots,\phi(T_{i})$. This can be done since each of the $i$ fixed colors conflicts with at most $d_{\mathcal{R}}(\mathcal{F})$ colors in $L_{i+1}$, i.e., there are at most $i\cdot d_{\mathcal{R}}(\mathcal{F})\leq m\cdot |\mathcal{F}|\cdot d_{\mathcal{R}}(\mathcal{F})$  colors  that $T_{i+1}$ cannot take, and this is less than the size of $L_{i+1}$, as assumed. 
\end{proof}

\subsection{Basic Problems: $P_0$ and $P_1$}
\label{sec:notition}
 Let $\mathcal{C}$ be a fixed and globally known color space (which may depend on the graph $G$). An \emph{$i$-list} is a subset $L\subseteq \Pow^{(i)}(\mathcal{C})$; e.g., the initial color list  $L_v\subseteq \mathcal{C}$ of a vertex $v$ is a $0$-list. 
  Below, we introduce  two problems. Problem $P_0$ is the standard list coloring problem, which we would like to solve via \Cref{l:zeroTypes}. However, the Lemma may not apply, if the lists $L_v$ are not large enough.
 We then introduce problem $P_1$, with parameters $0<\tz\le \kz$, which is a \emph{low intersecting sublist selection} problem. On the one hand, $P_1$ can be reformulated as a conflict coloring problem with larger lists and color space (hence could be solvable via \Cref{l:zeroTypes}), and on the other hand, a solution to $P_1$ can be used to solve $P_0$.
The input of a node $v$ in both problems contains its list $L_v$.
 Formally, we have, for parameters $\tz$ and $\kz$, 
\begin{itemize}
\item \textbf{\boldmath $P_0$ (list coloring):} Node $v$ has to output a color $c(v)\in L_v$ such that adjacent nodes' colors do not conflict, i.e., they are not equal. 
\item \textbf{\boldmath $P_1$ (low intersecting sublists):} Node $v$ has to output a $0$-list  $C_v\subseteq L_v$  such that $|C_v|=\kz$ and adjacent nodes' $0$-lists do not $\tz$-conflict.
 
Two $0$-lists $C,C'\subseteq \mathcal{C}$  do \emph{$\tz$-conflict} if $|C\cap C'|\geq \tz$.
\end{itemize}
Note that problems  $P_0$ and $P_1$ are not conflict coloring problems in the formal sense defined above (e.g., we do not define a list family $\mathcal{F}$ or a list assignment $\mathcal{L}:V\rightarrow \mathcal{F}$).  The aim with such definitions is to have a higher level and more intuitive (but still formal) problem statement. 
As the name suggests, in $P_1$ each node needs to compute a subset of its list such that the outputs form a low intersecting set family. 
$P_1$ can be reduced to a formal conflict coloring problem $\PP_1$ whose solution immediately solves the $P_1$ instance (see Thm.~\ref{thm:linialReproven}).

\subsection{Warmup: Linial's Color Reduction (without Lists)}
\label{sec:Linial} 
As a demonstration, we use the introduced framework to re-prove Linial's color reduction theorem~\cite[Thm. 4.1]{linial92} (which was extended to directed graphs in~\cite{BE10sublog}).
An \emph{$r$-cover-free family} of size $k$ over a set $U$ is a collection of $k$ subsets $C_1,\ldots, C_k\subseteq U$ such that no set $C_i$ is a subset of the union of $r$ others. The obtained algorithm is essentially a \emph{greedy construction} (via Lemma~\ref{l:zeroTypes}) of an $r$-cover-free family  (with appropriate parameters) whose existence was proved via the probabilistic method in~\cite{linial92}.  This greedy construction was first obtained in~\cite{sos} but, to our knowledge, remained  unnoticed in the distributed computing community. 

While our aim is to make the proof below reusable for the later sections (hence the general statement in list coloring terms), we note that similar ideas can be used to obtain a less technical proof of Linial's color reduction (see \Cref{app:Linial}).
\begin{theorem}[\cite{linial92,BE10sublog}]
\label{thm:linialReproven}
Let the graph $G$ be $m$-colored and oriented, with max outdegree $\beta$. All nodes have an identical color list $L_v=\mathcal{C}$  from a color space $\mathcal{C}$. 
Further, $\mathcal{C}$, $m$, and $\beta$ are globally known. If 
$
|L_v|= l_0\ge 2e\cdot  \beta^2\cdot \lceil\log m\rceil
$
holds then in $1$ round in \LOCAL, each node can output a color from its list s.t. adjacent nodes output distinct colors. 
\end{theorem}
\begin{proof}
 Our goal is to solve $P_0$ with the given lists. To this end, it suffices to solve $P_1$ with parameters $\tau=\lceil \log m\rceil>1$ and $k=\beta\cdot \tau$, without communication. Indeed, having selected the sublists $(C_v)_{v\in V}$, we need only one round to solve $P_0$: Node $v$ learns sublists of its outneighbors and outputs any color $c(v) \in C_v \setminus \cup_{u\in N_{out}(v)} C_u$.
This can be done since for any outneighbor $u$, $C_v$ and $C_u$ do not $\tz$-conflict ($|C_v\cap C_u|\le \tz-1$) and hence $|C_v\setminus \cup_{u\in N_{out}(v)} C_u| \ge k-(\tz-1)\beta>0$.

\smallskip

We recast the given $P_1$ instance with (identical) input lists $(L_v)_{v\in V}$ as a conflict coloring problem $\PP_1=(\mathcal{C}_1, \mathcal{R}_1, \mathcal{F}_1,\mathcal{L}_1)$ with color space $\mathcal{C}_1=\Pow(\mathcal{C})$ and the $\tz$-conflict relation as $\mathcal{R}_1$. The list of a node $\mathcal{L}_1(v)={L_v\choose k}$ in $\PP_1$ consists of all $k$-sized subsets of its input list $L_v$. 
As each $L_v$ is identical to $\mathcal{C}$, we have that  $\mathcal{L}_1$ maps each node $v$ to the same list ${L_v\choose k}={\mathcal{C}\choose k}$ and the list family  $\mathcal{F}_1=\{{\mathcal{C}\choose k}\}$ consists of that singleton.  A solution to $\PP_1$ immediately solves $P_1$.

\begin{claim}
\label{claim:p1conflictDegree}
 The conflict degree of $\PP_1$ is upper bounded by $d_1 = {\kz \choose \tz}\cdot {{l_0-\tz} \choose {\kz-\tz}}$. 
\end{claim}
\begin{proof} 
Consider two arbitrary lists $L,L'\in \mathcal{F}_1$ and some $0$-list $C\in L$ (that is of size $k$). Each 0-list $C'\in L'$ that $\tz$-conflicts with $C$ can be constructed by first choosing $\tz$ elements of $C$, and then adding $k-\tz$ elements from the rest of $\mathcal{C}$ which is of size $l_0-\tau$. This can be done in at most $d_1$ many ways.
\renewcommand{\qed}{\ensuremath{\hfill\blacksquare}}
\end{proof}
\renewcommand{\qed}{\hfill \ensuremath{\Box}}
\begin{claim}\label{c:lonedone}Let $l_1={l_0\choose \kz}$. For any $k\ge \tz>1$, if $l_0\ge 2e\kz^2/\tz$ then $l_1/d_1> 2^\tau$.  \end{claim}
\begin{proof} We have
\begin{equation}
\frac{l_1}{d_1} = \frac{{l_0\choose \kz}}{{l_0-\tz\choose \kz-\tz}{\kz\choose \tz}}
>\left(\frac{l_0}{k}\right)^{\tau}\cdot \left(\frac{\tau}{ek}\right)^{\tau} =\left(\frac{l_0\tau}{ek^2}\right)^\tau\ge 2^\tau~,\label{eqn:lonedone}
\end{equation}
where in the first inequality, we used the well-known approximation ${\kz\choose \tz}\leq (e\kz/\tz)^\tz$, and 
the following inequality,  applied to ${l_0\choose \kz}/{l_0-\tz\choose \kz-\tz}$:  
for integers  $L>K>x>0$,  
\begin{equation}\label{eqn:A}
\frac{{L\choose K}}{{L-x\choose K-x}} = \frac{L!(K-x)!(L-K)!}{K!(L-K)!(L-x)!}
=\frac{L(L-1)\dots(L-x+1) }{K(K-1)\dots(K-x+1)}>\left(\frac{L}{K}\right)^{x}~, 
\end{equation}
which follows as $(L-i)/(K-i)>L/K$ holds for $0< i\leq x$.  
\renewcommand{\qed}{\ensuremath{\hfill\blacksquare}}
\end{proof}
\renewcommand{\qed}{\hfill \ensuremath{\Box}}
Since $\tau\ge \lceil\log m\rceil$ and $|\mathcal{F}_1|=1$ the last claim implies $|\mathcal{L}_1(v)| =l_1>m d_1\ge m |\mathcal{F}_1| d_{\mathcal{R}_1}(\mathcal{F}_1)$, hence we can solve $\PP_1$ (and thus also $P_1$) without communication, using  \Cref{l:zeroTypes}. 
\end{proof}

What we did above is  \emph{greedily} forming  a \emph{$\Delta$-cover-free family} $C_1,\ldots, C_m\subseteq [O(\Delta^2 \log m)]$ of \emph{size $m$}. 
The same was done in~\cite{linial92}, using the probabilistic method.
Having such a family globally known, every vertex of input color $x$ picks, in 0 rounds, the set $C_x$ as its candidate output colors (neighboring vertices get distinct sets). Then, every vertex of color $x$ learns the sets $C$ of its neighbors, and based on the $\Delta$-cover-free property and the fact that there are at most $\Delta$ neighbors, can select a color $c\in C_x$ that is not a candidate for any neighbor.

\section{Linial for Lists}
\label{sec:versionTwo}
The goal of this section is to prove the following theorem.

\smallskip 

\noindent\textbf{\Cref{thm:2rounds}} (Linial for Lists)\textbf{.}
\textit{
In a directed graph with max. degree $\Delta$, max. outdegree $\beta$, and an input $m$-coloring, list coloring with lists $L_v$ from a color space $\mathcal{C}$ and of size  \[|L_v|\geq l_0=4e\beta^2(4\log \beta+\log\log |\mathcal{C}|+\log\log m+8)\]  can be solved in 2 rounds in  \LOCAL. Each node sends $l_0\cdot \lceil\log|\mathcal{C}|\rceil+\lceil\log m\rceil$ bits in the first round and $\lceil l_0/4e\beta^2\rceil$ bits in the second.
}

\medskip

\textbf{Assumption.} Throughout this section, we assume that  \emph{the list of each node is exactly of size $l_0$}; if a node's list is larger it can select an arbitrary subset of size $l_0$. 

\subsection{The Problem $P_2$ (Low Intersecting Sublist Systems)} 
\label{ssec:sublistsystem}

Recall that when applying our framework to the case where all lists were equal (Thm.~\ref{thm:linialReproven}), 
we essentially constructed a $\Delta$-cover-free family over the color space, and this was sufficient because we only needed a family of size $m$: one set for each possible input pair $(x, L)$ of a color $x$ and list $L$, with  \emph{the same $L$ for all nodes}. 
In order to replicate the construction for list coloring, we would need to construct a cover-free family with a set $C_{x,L}$ for every combination of color $x$ and list $L$, and such that $C_{x,L}\subseteq L$. It is not hard to see that such a family does not exist. 
Instead, we introduce problem $P_2$, whose goal is to assign every input $(x,L)$ a \emph{collection of candidate subsets} $K_{x,L}=\{C_{x,L,1}, C_{x,L,2},\ldots \}$, where each $C_{x,L,i}\subseteq L$. 
Further, we need that for every pair of distinct collections, there are not many pairs of subsets from the two collections that intersect much (in a sense formally defined below). This ensures that having such $K_{x,L}$, the nodes can compute the desired $\Delta$-cover free family with one communication round, and use it to choose a color in another round.

Problem $P_2$ depends on  parameters $0<\tz\le\kz\le l_0$ and $0<\ta\le\ka$, and each node has a list $L_v\subseteq \mathcal{C}$ in its input.  Instead of a color (as in $P_0$) or a sublist (as in $P_1$), each vertex $v$ now needs to output a collection $K_v=\{C_1,C_{2},\ldots\}$ of sublists of $L_v$, each of size $\kz$.
\begin{itemize}
\item[] \textbf{\boldmath$P_2$ (low intersecting sublist systems):} Node $v$ has to output a $1$-list  $K_v\subseteq \Pow(L_v)$  s.t. adjacent nodes' $1$-lists do not $(\ta,\tz)$-conflict and $|K_v|=\ka$ and $|C|=\kz$ for all $C\in K_v$. 

Two $1$-lists $K, K'\subseteq \Pow(\mathcal{C})$ do \emph{$(\ta,\tz)$-conflict} 
if there are two sequences $C_1,\ldots,C_{\ta}\in K$ and $C'_1,\ldots,C'_{\ta}\in K'$, where at least one of the sequences has $\ta$ distinct elements and for every $1\leq i\leq \ta$, $C_i$ and $C_i'$ $\tz$-conflict.
\end{itemize} 
We  prove in \Cref{lem:wayup} that with a suitable choice of the parameters a solution of $P_2(\tz,\kz,\ta,\ka)$ yields a solution of $P_1(\tau,k)$ and $P_0$, where we implicitly impose (and throughout this section assume)  that $P_0, P_1$ and $P_2$ receive the same input lists $(L_v)_{v\in V}$.

\subsection{Algorithm}
\label{ssec:alg}
Under the assumptions of Thm.~\ref{thm:2rounds}, fix the following (globally known) four parameters: 
$\tz=\lceil 8\log\beta + 2\log\log|\mathcal{C}| + 2\log\log m\rceil+14$,  $\ta=2^{\tau-\lceil\log(2e\beta^2)\rceil}$,  $\kz=\beta\cdot \tz$ and $\ka=\beta\cdot\ta$.  
Note that, $\tau', k, k'$ are determined by $\tz$ and $\beta$. \footnote{It may also be helpful to note the similarity between this parameter setting and that in Thm.~\ref{thm:linialReproven}.}  We have the bound $l_0\ge 2e k^2/\tz$ on list size. 

\textbf{The algorithm}  consists of two phases. In the first phase, nodes \textbf{locally} and without any communication compute a solution $(K_v)_{v\in V}$ of $P_2$ consisting of $1$-lists  (see \Cref{l:newzero}).  The second phase has two rounds of communication. In the \textbf{first round}, each node $v$ learns the solution $K_u$ to $P_2$ of each outneighbor $u\in N_{out}(v)$, 
and selects a $0$-list $C_v\in K_v$ that does not conflict with the 0-lists in $K_u$, for $u\in N_{out}(v)$, and thus is a solution to $P_1$ (\Cref{lem:wayup}). In the \textbf{second round}, node $v$ learns the lists $C_u$ of outneighbors, and selects a color $c(v)\in C_v$ that does not appear in $C_u$, for $u\in N_{out}(v)$ (\Cref{lem:wayup}).  This solves $P_0$.

\begin{lemma}[$P_2\rightarrow P_1\rightarrow P_0$] \label{lem:wayup}
Given a solution $(K_v)_{v\in V}$ of $P_2$ (a solution $(C_v)_{v\in V}$ of $P_1$), a solution of $P_1$ (of $P_0$, resp.) can be computed in one round. 
\end{lemma}
\begin{proof}
\textbf{\boldmath $P_2\rightarrow P_1$:} As $K_v$ and $K_{u}$ do not $(\ta,\tz)$-conflict for any $u\in N_{out}(v)$, there are at most $\ta-1$ $0$-lists $C\in K_v$ that  $\tz$-conflict with a $0$-list in $K_u$. By removing all $C$ from $K_v$ that $\tz$-conflict with any $C'\in K_u$ for any outneighbor $u\in N_{out}(v)$ at least $|K_v|-\beta\cdot (\ta-1)=\ka-\beta\cdot (\ta-1)\geq 1$ outputs remain; let $C_v$ be any such $0$-list. 
As the conflict relation is symmetric, $P_1$ is solved.

\textbf{\boldmath $P_1\rightarrow P_0$:}
 Since $C_v,C_u$ do not $\tz$-conflict, removing from $C_v$ all the colors from the $0$-lists of the outneighbors leaves at least  $\kz-\beta\cdot (\tz-1)\geq 1$ colors  that $v$ can select as $c(v)$. 
\end{proof}

\subsection{Zero Round Solution to $P_2$}
\label{sec:zeroRound}
 The results in this section hold for parameters $\tz, \ta, \ka$ fixed as in \Cref{ssec:alg}, and for any $\tz\le \kz\le \beta\tz$. While we set $\kz=\beta\tz$ for solving $P_0$, we will use another value of $\kz$ for our defective coloring result (see \Cref{sec:listdefective}). 
Note that we still have the bound $l_0\ge 2e k^2/\tau$ on list size, for any such $\kz$. 
The goal of this section is to prove the following lemma.
\begin{lemma}[$P_2$ in zero rounds]\label{l:newzero}
Under the assumptions of Thm.~\ref{thm:2rounds}, 
the problem $P_2(\tz, \kz, \ta, \ka)$  can be solved in zero rounds. 
\end{lemma}
To prove \Cref{l:newzero}, we reduce (without communication) an instance of $P_2$ to a conflict coloring instance $\PP_2$ that can be solved in zero rounds with \Cref{l:zeroTypes}.

  \smallskip
\textbf{Reducing $P_2$ to a conflict coloring instance $\PP_2$ (without communication):}\\ 
Given input lists $(L_v)_{v\in V}$ and parameters $0< \tz\le \kz\le l_0$ and $0<\ta\le \kz$, the 
conflict coloring instance $\PP_2$ is given by  the colorspace  $\Pow^{(2)}(\mathcal{C})$, the $(\ta,\tz)$-conflict relation $\mathcal{R}_2$ on 1-lists, the list family $\mathcal{F}_2=\text{Im}(L_2)=\{L_2(S)\mid S\in {\mathcal{C} \choose l_0}\}$   and list $\mathcal{L}(v)=L_2(L_v)$ for node $v$, where $L_2: {\mathcal{C}\choose l_0}\rightarrow \Pow^{(3)}(\mathcal{C})$ maps $l_0$-sized subsets of $\mathcal{C}$ to $2$-lists and is defined below.   The map $L_2$, the colorspace, the conflict relation and the set family $\mathcal{F}_2$ are global knowledge and no communication is needed to compute the list $\mathcal{L}(v)$ of a node in $\PP_2$. 

To define the map $L_2$ we need another definition. 
For an integer $t\geq 0$ and a $2$-list $T$, a $1$-list $K\in T$ is \emph{$(T,t,\ta,\tz)$-good}
if there are less than $t$ $1$-lists $K'\in T$ such that $K$ and $K'$ do $(\ta,\tz)$-conflict.
 We define maps $L_1$, $\bar{L}_2$ and $L_2$, as follows. For  $S\in {\mathcal{C}\choose l_0}$,

\begin{align*}
L_1(S)&={S\choose \kz} && \text{(elements $C$ are $0$-lists)}\\
\bar{L}_2(S)& ={L_1(S)\choose \ka} && \text{(elements $K$ are $1$-lists)}\\
L_2(S)&=\{K\in \bar{L}_2(S) | \text{ $K$ is $(\bar{L}_2(S), d_2, \ta, \tz)$-good}\}  && \text{(elements $K$ are $1$-lists)}, 
\end{align*}
where $d_2$ is chosen as in \Cref{lem:l2Size}.\footnote{The precise value is not important to understand how $L_2$ is formed.} 
Due to the definition of the $(\ta,\tz)$-conflict relation and the map $L_2$, solving $\PP_2$ immediately solves $P_2$.

The sizes of $L_1(S), \bar L_2(S)$ and $L_2(S)$ do not depend on $S$. Let  $l_1=|L_1(S)|={l_0 \choose \kz}$, and $l_2=|\bar L_2(S)|/2={l_1 \choose k'}/2$.  We will later show that $|L_2(S)|\ge l_2$. Let $\mathcal{F}_1=\{L_1(S)\mid S\in {\mathcal{C} \choose l_0}\}$.

\textbf{Some intuition:} In the conflict coloring instance $\PP_2$, every node $v$ has a list $\{K_1, K_2, \ldots\}$ of $1$-lists, each a collection of subsets of its input list $L_v$. To ensure small conflict degree, but still large list size, it is enough that $v$ only takes $K$s that are ``good", as defined above. Since being ``good'' only depends on $L_v$, node $v$ can also compute its $\PP_2$-list locally.

In \Cref{lem:l2Size,lem:conflictDegree,lem:calc}, we show  that  lists $L_2(L_v)$ are large  and that $\PP_2$ has small conflict degree. Before that, let us see how these lemmas imply 0-round solvability of $P_2$ (\Cref{l:newzero}). 
\begin{proof}[Proof of \Cref{l:newzero}] To solve an instance of $P_2$ on input lists $(L_v)_{v\in V}$, nodes locally set up the conflict coloring instance $\PP_2$.  \Cref{lem:conflictDegree,lem:l2Size} show that the conflict degree  of $\PP_2$ is bounded by $d_{\mathcal{R}_2}(\mathcal{F}_2)\le d_2$, and that every list in $\mathcal{F}_2$ has size at least $l_2$. Note that  $\mathcal{F}_2$ is globally known and $|\mathcal{F}_2| = {|\mathcal{C}|\choose l_0}<|\mathcal{C}|^{l_0}$, since each element in $\mathcal{F}$ can be written as $L_2(S)$ for some $S\in {\mathcal{C}\choose l_0}$. Using \Cref{lem:calc} we obtain $l_2/d_2\geq \frac{1}{8}2^{2^{\tz-\log(4e\beta^2)}}\ge m\cdot |\mathcal{C}|^{l_0}>m\cdot |\mathcal{F}_2|$, where the second inequality follows by a routine calculation using the definition of $\tau$ and $l_0$ (see \Cref{app:calculationZero}).
Thus, \Cref{l:zeroTypes} holds, and $\PP_2$ and $P_2$ can be solved in zero rounds.
\end{proof}

We continue with proving  \Cref{lem:l2Size,lem:conflictDegree,lem:calc}. 
First, we bound the conflict degree of $\PP_2$.  Recall that it is a property of the list family $\mathcal{F}_2$ and the conflict relation $\mathcal{R}_2$, and is independent of the graph and list assignment. 
The proof involves establishing an isomorphism between $L_2(S)$ and $L_2(S')$, for any $S,S'$, which preserves their common elements. For this, it is crucial to have $|S|=|S'|$. This is why we need all input lists to have same size $|L_v|=l_0$.

\begin{lemma}[Conflict Degrees]
\label{lem:conflictDegree}
Let $X,Y\in {\mathcal{C}\choose l_0}$ be 0-lists. Let $d_1= {\kz \choose \tz}\cdot {{l_0-\tz} \choose {\kz-\tz}}$.
\begin{enumerate}
\item  For any $0$-list $C\in L_1(X)$, there are at most $d_1$ $0$-lists in $L_1(Y)$ that  $\tz$-conflict with $C$. 
\item  For any $1$-list $K\in L_2(X)$, there are at most $d_2$ $1$-lists in $L_2(Y)$ that  $(\ta,\tz)$-conflict with~$K$. In particular, $d_{\mathcal{R}_2}(\mathcal{F}_2)\le d_2$, and this holds irrespective of the value of $d_2$. 
\end{enumerate}
\end{lemma}
\begin{proof} The proof of the first claim is along the same lines as the proof of \Cref{claim:p1conflictDegree}, so
we only prove the second claim here.  Let $X_1=L_1(X),X_2=L_2(X)$ and $\bar X_2=\bar L_2(X)$, and define $Y_1,Y_2,\bar Y_2$ similarly. As $|X|=|Y|$, there is a bijection $\alpha: X\rightarrow Y$ that is the identity on $X\cap Y$: if $c\in X\cap Y$ then $\alpha(c)=c$. Further, since $X_1={X \choose k}$ and $Y_1={Y \choose k}$, we have the bijection $\beta:X_1\rightarrow Y_1$ given by $\beta(\{c_1,\dots,c_k\})=\{\alpha(c_1),\dots,\alpha(c_k)\}$, and since $\bar{X}_2={X_1 \choose \ka}$ and $\bar{Y}_2={Y_1 \choose \ka}$, we have the bijection $\gamma:\bar{X}_2\rightarrow \bar{Y}_2$, where $\gamma(\{C_1,\dots,C_{\ka}\})=\{\beta(C_1),\dots,\beta(C_{\ka})\}$.

\smallskip

We show that the claim holds for any $t\geq 0$ and for any $K\in \bar{X}_2$ that is $(\bar{X}_2,t,\ta,\tz)$-good (which demonstrates that the actual value of $d_2$ is irrelevant).  As  $Y_2\subseteq \bar{Y}_2$, it suffices to show that $K$ does  $(\ta,\tz)$-conflict with at most $t$ 1-lists in $\bar{Y}_2$.
Towards a contradiction, let   $K\in \bar{X}_2$ $(\ta,\tz)$-conflict with each of $t$ distinct 1-lists $K'_1,K'_2,\dots,K'_t\in \bar{Y}_2$ and define $K_i=\gamma^{-1}(K'_i)\in \bar{X}_2$.  We show that $K$ also  $(\ta,\tz)$-conflicts  with each of the distinct ($\gamma$ is a bijection) $K_1,\ldots, K_{t}\in \bar{X}_2$, which is a contradiction to $K$ being $(\bar{X}_2,t,\ta,\tz)$-good: 
To ease notation, let us focus on $K$ and $K_1$.
Assume there are $\ta$ distinct (case 2: not necessarily distinct) 0-lists $C'_1,C'_2,\dots,C'_{\ta}$ in $K'_1$, and $\ta$ not necessarily distinct (case 2: distinct) 0-lists $C_1,C_2,\dots,C_{\ta}$ in $K$, such that $C_i$ and $C'_i$ $\tz$-conflict. Then $\beta^{-1}(C'_i)$ and $C_i$ $\tz$-conflict, since $\alpha$ is the identity on $C_i\cap C'_i$, $\beta^{-1}(C'_i)$ are all distinct (since $\beta$ is a bijection) and belong to $K_1$, therefore $K$ and $K_1$ $(\ta,\tz)$-conflict. 
\end{proof}

Next, we show that at most half of the elements $K\in \bar{L}_2$ fail to be good; this lemma crucially depends on the value of $d_2$. Below, we use the conflict degree $d_1$ from Lemma~\ref{lem:conflictDegree}.
\begin{lemma}[$L_2$ is large]
\label{lem:l2Size}
Let $d_2=4{\ka d_1 \choose \ta}\cdot {l_1-\ta \choose \ka-\ta}$. For any $S\in {\mathcal{C}\choose l_0}$, we have $|L_2(S)|\geq l_2$.
\end{lemma}
\begin{proof}
Fix $S\in {\mathcal{C}\choose l_0}$ and consider the digraph $H=(V_H,E_H)$ over the vertex set $V_H=\bar{L}_2(S)$, where $(K,K')\in E_H$  iff $K$ contains at least $\ta$ lists, each in $\tz$-conflict with a list in $K'$ (in particular, for every $K$, $(K,K)\in E_H$). Note  that a 1-list $K$ is $(\bar{L}_2(S),d_2,\ta,\tz)$-good iff its undirected degree in $H$ is at most $d_2$.
\begin{claim*}
The maximum outdegree of a node $K\in V_H$ is at most $d_2/4$.
\end{claim*}
\begin{proof}
Consider a fixed $K\in V_H$. Let $X\subseteq K$ be the set of 0-lists in $L_1(S)$ that $\tz$-conflict with a 0-list in $K$. By \Cref{lem:conflictDegree} part 1, every $C\in K$ $\tz$-conflicts with at most $d_1$ of 0-lists, hence $|X|\le |K|\cdot d_1=\kz d_1$. Every 1-list $K'$, such that there are at least  $\ta$ 0-lists in $K$ that are in $\tz$-conflict with a 0-list in $K'$, can be obtained by first choosing $\ta$ 0-lists from $X$, and adding  an arbitrary subset of $\ka-\ta$ other 0-lists. Clearly, this can be done in at most ${\ka d_1 \choose \ta}\cdot {l_1-\ta \choose \ka-\ta}=d_2/4$ many ways.
\renewcommand{\qed}{\ensuremath{\hfill\blacksquare}}
\end{proof}
\renewcommand{\qed}{\hfill \ensuremath{\Box}}

The Claim implies that $|E_H|\le |V_H|\cdot d_2/4$, hence the undirected average degree of a node in $H$ is at most $2|E_H|/|V_H| \le d_2/2$, and by Markov's inequality, at most half of the nodes have degree greater than $d_2$. Since $L_2(S)$ is the set of nodes of degree at most $d_2$, we conclude that $|L_2(S)|\ge |V_H|/2=|\bar{L}_2(S)|/2=l_2$.
\end{proof}

Finally, we bound the ratio $l_2/d_2$ based on the values of the remaining parameters. 
\begin{lemma}[$l/d$ Ratio]
\label{lem:calc} If $k\geq \tau\geq \lceil\log(2e\beta^2)\rceil$,
$l_0\geq 2ek^2/\tau$, $\ta=2^{\tau-\lceil\log(2e\beta^2)\rceil}$, and $\ka=\beta \ta$,  then 
  $l_2/d_2> 2^{2^{\tau-\log(4e\beta^2)}}/8$.
\end{lemma}
\begin{proof}
First, we get $l_1/d_1\geq 2^{\tz}$, as in Eq.~(\ref{eqn:lonedone}). Then, with ${\ka d_1\choose \ta}\leq (\frac{ekd_1}{\ta})^{\ta}$, and (\ref{eqn:A}) applied to ${l_1\choose \ka}/{l_1-\ta\choose \ka-\ta}$, we lower bound $l_2/d_2$ as
\begin{align*}
\frac{l_2}{d_2}& =\frac{1}{8} \frac{{l_1\choose \ka}}{{l_1-\ta\choose \ka-\ta}{\ka d_1\choose \ta}} 
>\frac{1}{8}\left(\frac{l_1}{\ka}\cdot\frac{\ta}{e(\ka d_1)} \right)^{\ta}
\geq\frac{1}{8}\left(\frac{2^{\tz}}{e\beta^2\ta}\right)^{\ta}{\ge}\frac{2^{\ta}}{8}\geq \frac{2^{2^{\tau-\log(4e\beta^2)}}}{8}\ ,
\end{align*}
where the third and fourth inequalities hold since $\ta\le \frac{2^\tz}{2e\beta^2}\le 2\ta$. 
\end{proof}

\subsection{Proof of the Main Theorem}

\begin{proof}[Proof of Thm.~\ref{thm:2rounds}]
Nodes solve $P_2$ in zero rounds (\Cref{l:newzero}), and then use two rounds of communication to solve the input list coloring problem $P_0$ (see algorithm description and \Cref{lem:wayup}).
We bound the messages sent by a node $v$ during the algorithm. In the first round, $v$ needs to send $K_v$ to its neighbors. Note that $K_v$ is uniquely determined by the list $L_v$ and the input color $\psi_v$  (see the proof of \Cref{l:zeroTypes}) , so it suffices to send $(\psi_v,L_v)$, which can be encoded in $l_0\lceil\log|\mathcal{C}|\rceil+\lceil\log m\rceil$ bits. In the second round, $v$ needs to send $C_v$. Since $C_v\in K_v$, and the neighbors know $K_v$, it suffices to send the index of $C_v$ in $K_v$ (in a fixed ordering).  
 Recall that $|K_v|=\ka<2^\tz$, so $v$ only needs to send $\tz\le l/4e\beta^2$ bits. 
\end{proof}
\begin{remark}
Note that in both communication rounds of \Cref{thm:2rounds} each node only needs to send messages to its in-neighbors. In contrast, the results in \Cref{sec:coloring} and \Cref{sec:listdefective} require bi-directional communication. 
\end{remark}
\section{Application: $(\Delta+1)$-Coloring and $(deg+1)$-List Coloring}
\label{sec:coloring}

\noindent\textbf{\Cref{thm:mainListColoring} (Restatement).} 
\textit{
In a graph with max. degree $\Delta$, $(deg+1)$-list coloring with lists $L_v\subseteq \mathcal{C}$ from a color space  of size $|\mathcal{C}|=2^{\poly(\Delta)}$ can be solved in $O(\sqrt{\Delta\log\Delta})+\frac{1}{2}\cdot\logstar n$ rounds in \LOCAL.
Furthermore, each node only needs to broadcast to its neighbors a single non-\CONGEST message consisting of a subset of its list. 
}
\medskip

The proof combines Thm.~\ref{thm:2rounds} with the graph partitioning provided by \cite{BEG18}, following the high level description in Sec.~\ref{sec:intro}.  
A variant of this framework was also used in~\cite{Kuhn20,BKO20}. We nevertheless present a proof for completeness, and also due to subtle but important differences from \cite{FHK} (we have an additional finishing phase that is not present there).

The graph partitioning given by~\cite{BEG18} aims at \emph{arbdefective colorings}, as introduced in \cite{barenboimE10}, but the main technical object provided by~\cite{BEG18} (and which is all we need here) is a \emph{low outdegree partition} of a graph. For a graph $H=(V,E)$, the collection $H_1,\dots,H_k$ of \emph{directed} graphs $H_i=(V_i,E_i)$ is a \emph{$\beta$-outdegree partition} of $H$ if their vertices span $V$, i.e., $V=V_1\cup\dots\cup V_k$, the underlying undirected graph of $H_i$ is the induced subgraph $H[V_i]$ (so it is indeed a partition), and the max. outdegree of a node in $H_i$ is at most $\beta$, for all $i$.

\begin{lemma}[Lemmas 6.1-6.3, \cite{BEG18}]
\label{lem:lowOutDegree} There are constants $c,c'>0$, s.t. for every $\beta\ge c$, given a graph $H$ with an $m$-coloring, 
 there is a deterministic algorithm that computes a $\beta$-outdegree partition $H_1,\dots,H_k$ with $k=c'\Delta/\beta$ in $O(k+\logstar m)$ rounds in \CONGEST.
 \end{lemma}

\begin{proof}[Proof of Thm.~\ref{thm:mainListColoring}]
 We begin with computing an $m=O(\Delta^2)$-coloring in $\frac{1}{2}\logstar n+O(1)$ rounds \cite{linial92,SzegedyV93}. 
The main algorithm consists of $t=\log_2(\Delta/\Delta^{1/4})$ phases. After phase $j$, we have colored a subset of vertices, s.t. the maximum degree $\Delta_j$ of the graph $G[U_j]$ induced by uncolored vertices is upper bounded as $\Delta_j\le\Delta/2^j$.  
Before describing a phase $j$, let us show how we finish the coloring after the phase $t$, in a {\bf final phase}. Consider  the graph $G[U_t]$ at the end of phase $t$. Note that it has maximum degree $\Delta_t=O(\Delta^{1/4})$. 
We compute an $m'$-coloring of $G[U_t]$ with $m'=O(\Delta_t^2)=O(\sqrt{\Delta})$ from the initial $m$-coloring in $O(1)$ rounds \cite{linial92}. In each of the final $m'$ rounds $i=1,\dots,m'$, vertices with color $i$ pick a color from their list not picked by a neighbor (can be done since $|L_v|>\Delta$ and no two neighbors pick simultaneously). The runtime of the final phase is $O(\sqrt{\Delta})$.

The following happens in {\bf phase} $j=1,\dots,t$. At the beginning of the phase, we have the set $U_{j-1}$ of uncolored vertices, where $U_0=V(G)$. Let $X=4e \cdot (4\log \Delta+\log\log |\mathcal{C}|+\log\log m+8)=O(\log \Delta)$ and for $j=0,\ldots,t-1$ let $\beta_j=\sqrt{\Delta_j/(2X)}$ and $k_j=c' \cdot \Delta_j/\beta_j$, 
where $c'$ is the constant in Lemma~\ref{lem:lowOutDegree}.\footnote{In order to apply the lemma, we need $\beta_j\ge c$. Since $\beta_j\ge\beta_t=\Omega(\sqrt{\sqrt{\Delta}/\log\Delta})$, $\beta_j\ge c$ holds if $\Delta$ is large enough. For $\Delta=O(1)$, Thm.~\ref{thm:mainListColoring} holds via a $O(\Delta)+1/2\log^*n$ round algorithm (see e.g.~\cite{BEG18}).}  
We partition $G[U_{j-1}]$ into $\beta_j$-outdegree subgraphs $H_1,H_2,\dots,H_{k_j}$,  
using Lemma~\ref{lem:lowOutDegree}. The phase consists of $k_j$ stages $i=1,\dots,k_j$, each consisting of 3 rounds. 
In {\bf stage $i$}, we partially color $H_i$, as follows. For every uncolored vertex $v\in H_i$, let $L_{v,j,i}$ be the set of colors in $L_v$ that have not been taken by a neighbor of $v$. Let $W_i=\{v\in H_i : |L_{v,j,i}|\ge \beta_j^2X\}$. Color the graph $H_i[W_i]$ using Linial for Lists (Thm.~\ref{thm:2rounds}) with color space $\mathcal{C}$, the $\beta_j$-outdegree orientation and the $m$-coloring. This is a valid application of the theorem, by the definition of $X$, $\beta_j$ and $W_i$. In the third round of the stage, all nodes in $W_i$ send their color to their neighbors.
This completes the algorithm description. Clearly,  phase $j$ takes $3k_j$ rounds.

It remains to show that $\Delta_j\le \Delta/2^j$. We do this by induction, with base $j=0$, $\Delta_0=\Delta$. Assume $\Delta_j\le \Delta/2^j$ holds for some $j\ge 0$. Let $v\in U_j$ be a node that is uncolored at the end of phase $j$. We know that  $|L_{v,j,i}|<\beta_j^2X=\Delta_j/2$, in a stage $i$. Recall that $L_{v,j,i}$ is the set of colors in $L_v$  not taken by a neighbor of $v$. Since $|L_v|$ is larger than the number of neighbors of $v$, $|L_{v,j,i}|$ is  larger than the number of \emph{uncolored} neighbors of $v$.
Therefore $v$ has at most $|L_{v,j,i}|< \Delta_j/2\le \Delta/2^{j+1}$ neighbors in $U_j$, which proves the induction: $\Delta_{j+1}\le \Delta/2^{j+1}$.

Recall that $X=O(\log \Delta)$ and bound the runtime as follows:  
\[
\frac{1}{2}\log^* n +O(1) + \sum_{j=1}^{t} 3k_j+m'=\frac{1}{2}\log^* n +\sum_{j=1}^{t} 3c'\sqrt{\frac{X\Delta}{2^{j-1}}}+O(\sqrt{\Delta})=\frac{1}{2}\log^* n+O(\sqrt{\Delta\log\Delta})\ .
\]
 The second claim easily follows, recalling the message complexity of Linial for Lists.
\end{proof}
Note that the final phase in the algorithm above is necessary as otherwise, if the  recursion continued until the maximum degree of uncolored nodes was, say, $O(\log^{(3)}\Delta)$, their reduced list size would be similarly small, and we could no longer apply \Cref{thm:2rounds}, which requires lists of size $\Omega(\log\log |\mathcal{C}|)=\Omega(\log\log\Delta)$, as the color space does not change in the recursion. 

\begin{corollary}\label{cor:deltaplus1}
In a graph with max. degree $\Delta=\tilde{O}(\log n)$, $(deg+1)$-list coloring with lists $L_v\subseteq \mathcal{C}$ from a color space  of size $|\mathcal{C}|={\poly(\Delta)}$ can be solved in $\tilde{O}(\sqrt{\Delta})+\frac{1}{2}\cdot\logstar n$ rounds in \CONGEST.
\end{corollary}
\begin{proof}
In the algorithm of \Cref{thm:mainListColoring} each vertex only participates in a single instance of Linial for lists, and all other steps of the algorithm can be implemented in \CONGEST. Thus, as stated above, the only non-\CONGEST message $M$ of a node consists of a subset of its list. Note that we can always limit the lists to size $\Delta+1$, so the number of colors in $M$ is at most $\Delta+1$. Each color can be encoded in $O(\log|\mathcal{C}|)=O(\log \Delta)$ bits. Thus, $M$ can be encoded in $O(\Delta\log\Delta)$ bits. By the assumption of the claim, we have $\Delta\le a\log n \cdot (\log\log n)^b$, for constants $a\ge 1,b\ge 0$. It follows that $\Delta\log\Delta\le a\log n \cdot (\log\Delta)^{b+1}$, as otherwise we would have $\Delta>a\log n (\log \Delta)^b>a\log n (\log \log n)^b$. Thus, each node sends at most $O(\Delta\log\Delta/\log n)=O((\log\Delta)^{b+1})$ messages more than in the \LOCAL algorithm, and the runtime increases by the corresponding factor.
\end{proof}

\section{Defective (List) Coloring}
\label{sec:listdefective}

A \emph{$d$-defective $c$-coloring} is a $c$-coloring where each vertex $v$ can have at most $d$ neighbors with the same color  as $v$. 
A \emph{$d$-defective list coloring} is a list coloring where each vertex $v$ can have at most $d$ neighbors with the same color  as $v$.
As proven in \cite{Kuhn2009WeakColoring}, one can compute a $d$-defective $O((\Delta/(d+1))^2\log m)$-coloring, given an $m$-coloring, in one round in \CONGEST. As a warm up to our ``list version",  we re-prove this result by adapting the proof of Thm.~\ref{thm:linialReproven}. 

\begin{theorem}[\cite{Kuhn2009WeakColoring}]
\label{thm:linialDefective}
Let $d\ge 0$ be an integer. In a graph with max. degree $\Delta>d$ and  an input $m$-coloring, $d$-defective coloring with $2e\cdot  \lceil\Delta/(d+1)\rceil^2\cdot \lceil\log_2 m\rceil$ colors can be computed in one round in \CONGEST.
\end{theorem}
\begin{proof}
We adapt the proof of Thm.~\ref{thm:linialReproven}.
Note that the proof holds when $G$ is not directed, and $\beta$ is replaced with $\Delta$. For a consistent notation, let $\beta=\Delta$.
 First, we solve problem $P_1(\tz,\kz)$ 
 with parameters $\tz=\lceil\log m\rceil$, $\kz=\lceil\beta/(d+1)\rceil\tz$. As  \Cref{c:lonedone} holds for any $\kz \ge \tz \ge \log m$ and $l_0\ge  2e\kz^2/\tz$  the corresponding $P_1(\tz,\kz)$ can be solved locally with list size $l_0\ge  2e\kz^2/\tz=2e\lceil\beta/(d+1)^2\rceil\tz$, which holds by our theorem assumption. Let $(C_v)_{v\in V}$ be the solution to $P_1$, with $|C_v|=\kz$. Let $v$ be a node. Since the solution is conflict free, we have $|C_v\cap C_u|< \tz$, for every neighbor of $v$. For a color $c\in C_v$, let $f(c)$ be the number of neighbors $u$ s.t. $c\in C_u$. 
It follows that $\sum_{c\in C_v}f(c)<\beta\tz$, hence there is a color $c\in C_v$ with $f(c)<\beta\tz/|C_v|\le d+1$. Each node $v$ picks a color $c\in C_v$ with minimum $f(c)$.
\end{proof}
Note that, as for Thm.~\ref{thm:linialReproven}, the set systems in the proof of Thm.~\ref{thm:linialDefective} are computed greedily. The theorem also works with a $d_1$-defective $m$-coloring as input. Then, the defect of the output coloring is $d+d_1$. By iteratively applying this result and combining it with another set system construction based on polynomials over finite fields, one can compute a $d$-defective $O((\Delta/(d+1))^2)$-coloring in $O(\logstar m)$ rounds from a given $m$-coloring~\cite{Kuhn2009WeakColoring}, \cite[Section 3.2]{BarenboimEK14}.

\begin{theorem}[Defective List Coloring]
\label{thm:listDefective}
Let $d\ge 0$ be an integer. In a graph with max. degree $\Delta>d$ and  an input $m$-coloring, $d$-defective list coloring with lists $L_v$ from a color space $\mathcal{C}$ and of size  $|L_v|\geq l=4e\lceil\Delta/(d+1)\rceil^2 \cdot (4\log\Delta + \log\log|\mathcal{C}| + \log\log m+8)$  can be solved in 2 rounds in \LOCAL, if $\mathcal{C}$, $m$ and $\Delta$  are globally known.
\end{theorem}

\begin{proof} As in Thm.~\ref{thm:linialDefective}, observe that our whole analysis of Linial for lists  holds in the case when $G$ is not directed, and $\beta$ is replaced with $\Delta$. Let $\beta=\Delta$, as before.
First, we solve problems $P_2$ and $P_1$ with parameters $\tz,\ta,\ka$ as in Sec.~\ref{sec:zeroRound}, and $\kz=\lceil\beta/(d+1)\rceil\tz$ (as observed in the beginning of Sec.~\ref{sec:zeroRound}, $P_2$ and $P_1$ can be solved for any $\tz\le \kz\le \beta\tz$). 
Thus, we require the list size of each node to be at least $l_0\ge 2e\kz^2/\tz$, which holds. Given a solution of $P_1$, we can obtain the $d$-defective coloring in the same manner as in~Thm.~\ref{thm:linialDefective}. 
\end{proof}

\section{Discussion}
\label{sec:discussion}
We conclude with several observations on our results, as well as open problems.
\begin{enumerate}
\item
It is possible to define problems $P_3,\ldots, P_t$ for any $t$, as we defined $P_1$ and $P_2$. \Cref{lem:conflictDegree} extends naturally to these problems, so the input of a node $v$ in $P_i$ is again only its initial list $L_v$. 
We  need $t$ rounds, instead of 2, to derive a solution of $P_0$ from a solution of $P_t$ (which also implies larger messages). On the other hand, we  have somewhat smaller list size requirement: $c\beta^2(\log \beta + \log^{(t)} |\mathcal{C}|+\log^{(t)} m)$, for a constant $c>0$.
In particular, one can list color in $O(\logstar \max\{|\mathcal{C}|,m\})$ rounds if lists are at least $c\beta^2\log\beta$ for a sufficiently large constant $c>0$. 
\item 
Unlike in~\cite{FHK}, our bound on the list size does not depend on $\Delta$. This result implies that, e.g., given a graph with a $\beta$-outdegree orientation and an input coloring with $2^{\poly(\beta)}$ colors and list sizes of at least $c\beta^2\cdot \log \beta$ from a color space of size $2^{\poly(\beta)}$, for a constant $c>0$, it is possible to list-color the graph in $2$ rounds. By the remark above, one can have even larger color space, by increasing the runtime accordingly.
\item A lower bound in~\cite{SzegedyV93} suggests that the coloring in \Cref{thm:2rounds} cannot be done in a single round. In particular, if one is willing to keep the doubly-logarithmic dependence on $m$ in the list size, then one has to pay a factor  exponential in $\beta$. On the other hand, we do not know how to eliminate the $\log\beta$ term, even if we use more communication.
\item The recently popular \emph{speedup} technique has mostly been used to prove lower bounds, e.g.,
 \cite{Brandt19speedup,brandt2016LLL,FOCS19MIS,br2020truly,balliu2020ruling,balliu2019classification}. Here, a problem $P_0$ is mechanically (and without communication!) transformed into a problem $P_1$ whose complexity is exactly one round less. Then, if $P_1$ cannot be solved locally one deduces that $P_0$ cannot be solved in $1$ round. By iterating this process, one can derive larger lower bounds. However, the description complexity of derived problems 
 grows exponentially, and it is very important to be able to simplify the problem description, in order to iterate the process. If $P_0$ is the $(\Delta+1)$-vertex coloring problem, this process has only been understood in the special case of $\Delta=2$, which corresponds to Linial's $\Omega(\logstar n)$ lower bound~\cite{linial92,Suomela14}. 
 While \cite{FHK} also performs a similar transformation, it is different from the speedup technique, since the transformation is not mechanical, requiring nodes to communicate for building the new problems. It may rather be seen as a transformation of \emph{problem instances} (that depend on the graph) than problems. In contrast, our transformations are mechanical, and the input  and  output labels live in the same universe as it is the case for mechanical speedup.
 \item {While our present treatment of the proof of Thm.~\ref{thm:2rounds} in terms of conflict coloring problems and problems $P_0, P_1$ and $P_2$ has the aim of connecting to the framework of~\cite{FHK} as well as to the speedup framework, we note that the proof can be stated entirely in terms of set systems, just like the proof of Linial's color reduction.}
\end{enumerate}

\begin{itemize}
\item  {\bf Open Problem:} Remove the $\log\beta$ term in Thm.~\ref{thm:2rounds} while keeping the runtime $o(\sqrt{\log \beta})$.
\end{itemize}
This question is particularly of interest because $\log\beta$ is the source of the $\sqrt{\log \Delta}$ factor in  in Thm.~\ref{thm:mainListColoring} (note that the terms depending on $m,|\mathcal{C}|$ can be reduced, by the remarks above). The non-list $O(\Delta^2)$-coloring by Linial uses, in addition to his main color reduction, a $O(\Delta^3)$-to-$O(\Delta^2)$ color reduction, using polynomials over finite fields \cite{linial92}. With a more sophisticated use of polynomials \cite{barenboim16} constructs a cover-free family for list coloring but it requires a much smaller outdegree. It is not clear if polynomials help with our question. 
\begin{itemize}
\item  {\bf Open Problem:} More generally, prove or rule out a truly local $(\Delta+1)$-coloring algorithm with $\Delta$-dependence $f(\Delta)=o(\sqrt{\Delta})$.
\end{itemize}

\section*{Acknowledgements}
This project was supported by the European Union's Horizon 2020 Research and  Innovation Programme under grant agreement no. 755839.
\clearpage
\bibliographystyle{alpha}
\bibliography{references}
\clearpage
\appendix

\section{Color Reduction via Greedy Construction of Cover-Free Families}
\label{app:Linial}

For an integer $\Delta\ge 1$, a \emph{$\Delta$-cover free} family $\mathcal{F}\subseteq 2^U$ over a universe $U$ is a collection of subsets of $U$, such that no set in $\mathcal{F}$ is contained in the union of $\Delta$ others.

To be self contained we quickly repeat how \cite{linial92} use such families to reduce a coloring: Suppose a graph $G$ of max. degree $\Delta$ is $m$-colored and there is a globally known $\Delta$-cover free family $\mathcal{F}=\{C_1,C_2\dots\}$ with $|\mathcal{F}|\ge m$ over a universe $U=[m']$.  Then $G$ can be recolored with $m'$ colors, as follows. Every node $v$ of input color $x$ selects $C(v)=C_x\subseteq U$, sends $x$ to its neighbors, after which every node picks a color $y\in C(v)\setminus \cup_{u\in N(v)}C(u)$: $y$ exists since $|N(v)|\le \Delta$ and $\mathcal{F}$ is $\Delta$-cover free. By the choice of $y$, the coloring is proper. 
\begin{theorem}
For integers $\Delta\ge 2$ and $m\ge 3$, there is a greedy algorithm that constructs a $\Delta$-cover free set family $\mathcal{F}$ of size $m$ over a universe of size at most $5.2\Delta^2\log_2(em)$.
\end{theorem}
\begin{proof}
Let  $z=\lceil \ln m\rceil$, $x=\Delta z$ and $l=\lceil e^2x/z\rceil$, and note that $z\le x\le l$. Let $\mathcal{P}$ be the set of all functions from $[x]$
to $[l]$. Two functions \emph{conflict} if they agree on at least $z$ inputs. Consider the following greedy procedure for selecting a \emph{conflict-free} set $\mathcal{R}$ of functions: 
\begin{itemize}
\item  $\mathcal{R}\gets \emptyset$, $\mathcal{T}\gets \mathcal{P}$, 
\item \textbf{while} $\mathcal{T}\neq \emptyset$
\begin{itemize}
\item $\mathcal{R}\gets \mathcal{R}\cup \{f\}$ for arbitrary $f\in \mathcal{T}$
\item $\mathcal{T}\gets \mathcal{T}\setminus\{g\in \mathcal{T}\mid g \text{ conflicts with } f\}$. 
\end{itemize}
\end{itemize}


\noindent\textbf{No pair of functions in $\mathcal{R}$ conflict:} 
After each iteration of the while-loop the set $\mathcal{T}$ only contains functions that do not conflict with any function in $\mathcal{R}$. Thus, in each iteration the function $f\in \mathcal{T}$ that we add to $\mathcal{R}$ does not conflict with any previously picked function in $\mathcal{R}$. 

\smallskip

\noindent\textbf{\boldmath$|\mathcal{R}|\geq m$:} For a given function $f$, there are at most $d={x \choose z} l^{x-z}$ functions in $\mathcal{P}$ that conflict with $f$: every such function can be obtained by making it agree with $f$ at some $z$ inputs, and choosing the value on the rest of the inputs arbitrarily.
%
Therefore, the number of functions removed from $\mathcal{T}$ at every step is at most $d$. Thus, the process runs for at least $|\mathcal{P}|/d$ iterations and using ${x \choose z}\leq (ex/z)^z$, $l\ge e^2x/z$, and $z\ge \ln m$ we obtain
$|\mathcal{R}|\ge \frac{|\mathcal{P}|}{d}\geq \left(\frac{l\cdot z}{e\cdot x}\right)^z\geq  e^z\ge m$~. 

\smallskip

Now, the $\Delta$-cover-free family $\mathcal{F}$ over the universe $U=[x]\times [l]$ consists of sets $S_f=\{(i,f(i))\mid i\in [x]\}\subseteq U$, for all $f\in \mathcal{R}$.  Let $f\neq g\in \mathcal{R}$. Since $f$ and $g$ do not conflict, there are at most $z-1$ values $i\in [x]$, s.t. $f(i)=g(i)$. Therefore, $|S_f\cap S_g|\le z-1$. Since $|S_f|=x=z\Delta$, for every $\Delta$ distinct functions $g_1,\dots,g_\Delta\in \mathcal{R}\setminus \{f\}$, 
 we have $|S_f\setminus \cup_{t=1}^\Delta S_{g_t}|\ge x-\Delta(z-1)=\Delta$.
The size of the universe is $|U|=x\cdot l\le7.4\Delta^2z\le  7.4\Delta^2 \ln (em)\le 5.2\Delta^2\log_2(em)$.
\end{proof}

\section{Detailed Calculation for the Proof of \Cref{l:newzero}}
\label{app:calculationZero}
Recall that we need to show that for $\tz=\lceil 8\log\beta + 2\log\log|\mathcal{C}| + 2\log\log m\rceil+14$, and $l_0= 2e\beta^2(\tau+1)$, it holds that $(1/8) \cdot 2^{2^{\tz-\log(4e\beta^2)}}\ge m|\mathcal{C}|^{l_0}$. We have: 
\begin{align*}
 \log \log \left(8m\cdot |\mathcal{C}|^{l_0}\right)&\leq \log\log m + \log l_0  +\log\log C +\log 3\\
& \leq  \log\log m + \log 2e+\log 3 + 2\log\beta +\log (\tz+1)  +\log\log C \ \\
& \leq \tz/2+ \log\log m + \log 2e+\log 3 + 2\log\beta   +\log\log C \\
& \leq \tz/2+\tz/2 -\log(4e\beta^2) =\tz-\log(4e\beta^2)\ ,
\end{align*}
where in the first inequality, we used $\log(x+y)\le \log x + \log y$, for $x,y\ge 1$, in the second one we used $l_0=2e\beta^2(\tz+1)$, in the third we used $\log (\tz+1) \leq \tz/2$ (holds for $\tz\geq 6$), while in the last one we used $\tz/2\ge 4\log\beta + \log\log |C| + \log\log m + 7$. 
 
Exponentiating (with base 2) both sides twice gives us the claim.
\end{document}